\documentclass[runningheads,envcountsame]{llncs}
\usepackage{amsmath,amssymb}
\usepackage{proof}
\usepackage{stmaryrd}
\usepackage{latexsym}
\usepackage{url}

\newcommand\eg{\hbox{\textit{e.g.}}}
\newcommand\ie{\hbox{\textit{i.e.}}}

\newcommand{\nat}{\mathbb{N}}

\newcommand\sv{\mathtt{SV}}
\long\def\hide#1\endhide{}

\newcommand\ra\rightarrow

\newcommand\iimp{\ra}

\newcommand\bottom{\perp}

\newcommand\tensor\otimes

\def\Fscr{{\cal F}}
\def\Gscr{{\cal G}}

\def\Mscr{{\cal M}}
\def\Nscr{{\cal N}}

\def\Rscr{{\cal R}}

\def\Vscr{{\cal V}}

\newcommand{\lns}[1][]{\mathbin{\!/\mkern-5mu/^{#1}\!}}
\newcommand{\lnse}[1][]{\mathbin{\!\backslash\mkern-5mu\backslash^{#1}\!}}

\DeclareSymbolFont{symbolsC}{U}{txsyc}{m}{n}
\DeclareMathSymbol{\boxright}{\mathrel}{symbolsC}{128}
\DeclareMathSymbol{\boxRight}{\mathrel}{symbolsC}{136}

\newcommand{\4}{\mathsf{4}}
\newcommand{\5}{\mathsf{5}}

\newcommand{\cless}{\preccurlyeq}

\newcommand{\cut}{\mathsf{cut}}

\newcommand{\cls}{\mathtt{close}}

\newcommand{\D}{\mathsf{D}}
\newcommand{\defs}{:=}

\newcommand{\dpt}[1]{\operatorname{\mathsf{dp}}\left( #1 \right)}

\newcommand{\e}{\mathsf{e}}

\newcommand{\forcef}{\Vdash^{\forall}}

\newcommand{\G}{\mathsf{G}}
\newcommand\Gtcp{\mathsf{G3cp}}
\newcommand\GtE{\mathsf{G3E}}

\newcommand\GtI{\mathsf{G3I}}
\newcommand\GtM{\mathsf{G3M}}
\newcommand\Gtmm{\mathsf{G3MM}}

\newcommand{\ibox}[1]{{\Box_{#1}}}
\newcommand{\id}{\mathtt{id}}

\newcommand{\init}{\mathsf{init}}

\newcommand{\K}{\mathsf{K}}
\newcommand{\lab}{\!:\!}

\newcommand{\lift}{\mathtt{lift}}

\newcommand{\lowe}{\mathtt{lower}}

\newcommand{\LL}{\mathsf{LL}}

\newcommand{\LLS}{\mathsf{LbLNS}}
\newcommand{\LLNS}{\mathsf{LbNS}}

\newcommand{\LNS}{\mathsf{LNS}}

\newcommand{\m}{\mathsf{m}}
\newcommand{\M}{\mathsf{M}}

\newcommand\mLJ{\mathsf{mLJ}}

\newcommand{\nes}[1]{\left[#1\right]}
\newcommand{\nese}[1]{\left[#1\right]^\e}

\newcommand{\NS}{\mathsf{NS}}

\newcommand{\SC}{\mathsf{SC}}
\newcommand{\seq}{\vdash}

\newcommand{\Sf}{\mathsf{S4}}
\newcommand{\Sfi}{\mathsf{S5}}

\newcommand{\T}{\mathsf{T}}
\newcommand{\TB}{\mathsf{TB}}

\newcommand{\upset}[2][]{{\uparrow^{#1}\!\!({#2})}}

\newcommand\world[2]{#1:#2}

\usepackage{color}
\usepackage{xspace}
\usepackage{relsize}
\usepackage{times}
\RequirePackage{txfonts}
\usepackage{microtype}
\usepackage[show]{ed}
\usepackage{hyperref}
\usepackage{graphicx}
\usepackage{subcaption}

\title{Proof systems: from nestings to sequents and back}
\author{Elaine Pimentel\inst{1}\thanks{Extended version. This work has been partially supported by CAPES, CNPq and  project FWF START Y544-N23.}}

\institute{
Departamento de Matem\'atica, UFRN, Brazil}

\begin{document}
\maketitle

\begin{abstract}

In this work, we explore proof theoretical connections between sequent, nested  and labelled calculi. In particular, we show
a general algorithm for transforming a class of nested systems into  sequent calculus systems, passing through linear nested systems. Moreover, we show a semantical characterisation  of intuitionistic, multi-modal and non-normal modal logics for all these systems, via a case-by-case translation between labelled nested to labelled sequent systems.  \end{abstract}

\section{Introduction}\label{sec:intro}

The quest of finding {\em good} proof systems for different logics has been the main research topic for proof theorists since Gentzen's seminal work~\cite{gentzen35}. The definition of good is, of course, subjective.
While it is widely accepted that a proposed calculus has to be sound and complete w.r.t. a given semantics, other aspects such as {\em analicity}, {\em simplicity}, and {\em efficiency} are often taken into account for considering a calculus ``adequate''.

One of the best known formalisms for proposing analytic proof systems is Gentzen's {\em sequent calculus}. While its simplicity makes it an ideal tool for proving meta-logical properties, sequent calculus is not expressive enough for constructing analytic calculi for many logics of interest. 
The case of modal logic is particularly problematic, since sequent systems for such logics  are
usually not modular, and they mostly lack relevant
properties such as separate left and right introduction rules for the
modalities.  These problems are often connected to the fact that the
modal rules in such calculi usually introduce more than one connective
at a time, e.g. as in the
rule $\mathsf{k}$ for modal logic $\K$:
\begin{center}
$
\infer[\mathsf{k}]{\Box B_1, \ldots, \Box B_n\seq\Box A}{B_1, \ldots, B_n\seq A} 
$
\end{center} 
One way of solving this problem is by considering extensions of the
sequent framework that are expressive enough for capturing these
modalities using separate left and right introduction rules.  This is
possible e.g. in \emph{labelled
  sequents}~\cite{DBLP:books/daglib/0003059} or in \emph{nested sequents}~\cite{Brunnler:2009kx}.
In the labelled sequent framework,
usually the semantical characterisation  
is explicitly added to sequents. In the nested
framework in contrast, a single sequent is replaced with
a tree of sequents, where successors of a sequent (nestings) are interpreted under a given modality. The nesting rules of these calculi govern the transfer
of  formulae between the different sequents, and 
they are  {\em local}, in the sense that it is sufficient to transfer only one formula at a
time. As an example, the labelled and nested versions for the {\em necessity right rule} ($\Box_R$) are
\begin{center}
$ \infer[\Box_R^l]{\mathcal{R},X\seq Y,x\lab \Box A}{\mathcal{R},xRy,X\seq Y,y\lab A} \qquad
           \infer[\Box_R^n]{\Gamma\vdash \Delta,\Box A}{\Gamma\vdash\Delta,[\cdot\vdash A]}
$
\end{center}
where $y$ is a fresh variable in the $\Box_R^l$ rule.
Reading bottom up, while the labelled system creates a new variable $y$ related to  $x$ via a relation $R$ and changes the label of $A$ to $y$, in $\Box_R^n$ a new nesting is created, and $A$ is moved there.
It seems clear that nestings and semantical structures are somehow related. 
Indeed, a direct translation between proofs in labelled and nested systems for 
some
normal modal logics is presented in~\cite{DBLP:conf/aiml/GoreR12}, while in~\cite{Fitting:2014} it is shown how to relate nestings with Kripke structures for intuitionistic logic.
In this work, we show how to smoothly generalise this  relationship to multi-modalities,
where intuitionistic logic and  normal modal logics are particular cases.

Since nested systems have being also proposed for other modalities, such as the non-normal ones~\cite{Chellas:1980fk}, an interesting question is 
whether this semantical interpretation can be generalised to other systems as well.
In~\cite{Negri2017} a labelled approach was used for setting the grounds for proof theory of some non-normal modal systems based on {\em neighbourhood semantics}. 
In parallel, we have proposed~\cite{DBLP:journals/corr/LellmannP17} modular systems based on  nestings for several non-normal modal logics. We will relate these two approaches for the logics $\M$ and $\mathsf{E}$, hence clarifying the nesting-semantics relationship for such logics. 

While nested sequents allow for modular proposal of proof systems, it comes with a price:
the obvious proof search procedure is of suboptimal complexity since
it constructs potentially exponentially large nested
sequents~\cite{Brunnler:2009kx}. In this work,
we show that a class of nested systems can be transformed into sequent 
systems via a linearisation procedure, where sequent rules can be seen as nested {\em macro-rules}. In this way, we do not only recover simplicity
but also efficiency by defining an optimal proof search procedure, in the sense that it matches the computational complexity of the validity problem in the logic.

Finally, by relating nested and sequent systems, we are able to extend the semantical interpretation also to the sequent case, hence closing the relationship between systems and shedding light on the semantical interpretation of several sequent based systems.

{\bf Organisation and contributions.} Sec.~\ref{sec:sc} presents the basic notation for sequent systems; Sec.~\ref{sec:nested} shows sufficient conditions for a nested system to be linearised, so that to be presented 
as linear a nested system; Sec~\ref{sec:lnested} shows how to transform linear nested systems into sequent systems; Sec~\ref{sec:lbns} presents the basic notation for labelled systems; Sec.~\ref{sec:int}, \ref{sec:mm} and 
\ref{sec:nn} show the results under the particular views of intuitionistic, multi-modal and non-normal logics; Sec.~\ref{sec:conc} concludes the paper.


\section{Sequent systems}\label{sec:sc}

Contemporary proof theory started with Gentzen's work~\cite{gentzen35}, and it has had a
continuous development with the proposal of several proof systems for
many logics.

\begin{definition}\label{def:seq}
  A {\em sequent} is an expression of the form $\Gamma\vdash \Delta$
  where $\Gamma$ (the \emph{antecedent}) and $\Delta$ (the
  \emph{succedent}) are finite multisets of formulae. A {\em sequent calculus ($\SC$)} consists of a set of rules, 
  of
  the form
  $$
  \infer[r]{S}{S_1\quad\cdots\quad S_n}
  $$
  where the sequent $S$ is the {\em conclusion} inferred from the {\em
    premise} sequents $S_1,\ldots,S_n$ in the rule $r$.  If the set of
  premises is empty, then $r$ is an {\em axiom}.
where the formula  in the
  conclusion sequent in which a rule is applied is the {\em principal formula}, and its
  sub-formulae in the premises are the {\em auxiliary formulae}.

A {\em derivation} is a finite directed tree
with nodes labelled by sequents and a single root, axioms at the top
nodes, and where each node is connected with the (immediate) successor
nodes 
(if any) 
according to the inference rules.
  The {\em height} of a derivation is the greatest number of
  successive applications of rules in it, where an axiom has height 0.
\end{definition}
As an example, 
Fig.~\ref{fig:mLJ} presents $\SC_\mLJ$~\cite{maehara54nmj}, a multiple conclusion sequent system for propositional intuitionistic logic.
\begin{figure}[t]
$
\begin{array} {cccc}
\infer[{\iimp_L}]{\Gamma, A\iimp B \seq \Delta}{\Gamma, A\iimp B
 \seq A, \Delta \quad \Gamma, B \seq
\Delta }
&
\infer[{\iimp_R}]{\Gamma \seq A \iimp B, \Delta}{\Gamma, A
\seq B}
&
\infer[{\land_L}]{\Gamma, A \land B \seq \Delta}{\Gamma,  A,B \seq \Delta }
&
\infer[{\land_R}]{\Gamma \seq A \land B, \Delta}{\Gamma \seq
 A, \Delta \quad \Gamma \seq B, \Delta}
\\
[5pt]
\infer[{\lor_L}]{\Gamma, A \lor B \seq  \Delta}{\Gamma, 
A, \seq \Delta \quad \Gamma, B \seq  \Delta}
&
\infer[{\lor_R}]{\Gamma \seq A \lor B, \Delta}{\Gamma \seq A,B, \Delta }
&
\infer[{\bot_L}]{\Gamma,\bot \seq \Delta}{} 
&
 \infer[{\init}]{\Gamma,P \seq P, \Delta}{} 
\end{array}
$
\caption{Multi-conclusion intuitionistic calculus $\SC_\mLJ$. $P$ is atomic.}\label{fig:mLJ}
\end{figure}
The rules are exactly the same as in classical logic,
except for the implication rules. While the left rule copies the implication in the left premise, the right implication forces all formulae in the succedent of the conclusion sequent to be weakened (when viewed bottom-up). This guarantees that, on applying the ($\iimp_R$) rule on $A\iimp B$, the formula $B$ should be proved assuming {\em only} the pre-existent antecedent context extended with the formula $A$, creating an interdependency between $A$ and $B$.

\section{Nested systems}\label{sec:nested}

Nested systems~\cite{Brunnler:2009kx,Poggiolesi:2009vn} are
 extensions of the sequent framework where a single sequent is replaced with
a tree of sequents. 
\begin{definition}\label{def:nested}
A {\em nested sequent} is defined inductively as follows: 
\begin{itemize}
\item[(i)] if $\Gamma\seq\Delta$  is a sequent, then it is a nested sequent;
\item[(ii)] if $\Gamma\seq\Delta$ is a sequent and $G_1,\ldots,G_n$ are nested sequents, then $\Gamma\seq\Delta,[G_1],\ldots,[G_n]$ is a nested sequent.
\end{itemize}
A {\em nested system  ($\NS$)} consists of a set of inference rules acting on
nested sequents.
An {\em auxiliary structure} is either an auxiliary formula or 
a nesting created by the principal formula of a nested rule.
\end{definition}
For readability, we will denote by $\Gamma,\Delta$ sequent contexts and by $\Lambda$ multisets of nestings.
In this way, every nested sequent  has the shape $\Gamma\seq\Delta, \Lambda$
where elements of $\Lambda$ have the shape 
$[\Gamma'\seq\Delta', \Lambda']$ and so on. We will denote by $\Upsilon$ an arbitrary nested sequent. 

Nodes in a nesting tree are called {\em positions}~\cite{Klop:2001:TRS:559395}.  
This terminology is  transferred to formulae, so that we also refer to the position of a formula in a nesting. We will index a node in a nesting  tree by a finite, possibly empty, sequence of positive integers
$i_1.i_2.\ldots.i_k$, indicating the path from the root of the term tree to the intended position. Such sequences will be used to denote the corresponding positions.  
Given two sequences $p,q$, we write $p \leq q$ if $q$ extends $p$, that is, if $q = p.p'$ for some sequence $p'$. We will denote by $p\parallel q$ if $p$ and $q$ are not comparable. Note that $p\leq q$ iff $q$ is in the subtree with $p$ as root.
We will denote by  $\id_\Upsilon(A)$ the position of a formula $A$ in a nested sequent $\Upsilon$. 
\begin{definition}\label{def:rd}
Let $A,B$ be formulae occurring in  nested sequents $\Upsilon,\Upsilon'$ respectively. Then 
$A\preceq B$ iff $\id_\Upsilon(A)\leq \id_{\Upsilon'}(B)$. 
\end{definition}
Rules in nested systems will be represented using {\em holed contexts}. 
\begin{definition}\label{def:holecont}
A {\em nested-holed context} is a nested sequent that contains a hole of the form $\{ \; \}$ in place of nestings. We represent such a context as $\Upsilon\{\;\}$. Given  a holed context and a nested sequent $\Upsilon'$, we write $\Upsilon\{\Upsilon'\}$ to stand for the nested sequent where the hole $\{\;\}$  has been replaced by $\nes{\Upsilon'}$, assuming that the hole is removed if $\Upsilon'$ is empty
and if $\Upsilon$ is empty then $\Upsilon\{\Upsilon'\}=\Upsilon'$. 
The depth of $\Upsilon\{\;\}$, denoted by $\dpt{\Upsilon\{\;\}}$, is the number of nodes on a branch of the nesting tree of $\Upsilon\{\;\}$ of maximal length.
\end{definition}
For example, $(\Gamma\seq\Delta,\{\;\})\{\Gamma'\seq\Delta'\}=
\Gamma\seq\Delta,\nes{\Gamma'\seq\Delta'}$ while 
$(\seq \{\;\})\{\Gamma'\seq\Delta'\}=
\Gamma'\seq\Delta'$. 
Fig.~\ref{fig:nestedmLJ} presents the  $\NS_{\mLJ}$~\cite{Fitting:2014}, a nested
system for $\mLJ$.
\begin{figure}[t]
\[
\infer[\iimp_L^n]{\Upsilon\{\Gamma, A \iimp B\seq\Delta,\Lambda\}}{\Upsilon\{\Gamma, A \iimp B\seq\Delta,A,\Lambda\}\quad
\Upsilon\{\Gamma,B\seq\Delta,\Lambda\}} \quad
\infer[\iimp_R^n]{\Upsilon\{\Gamma\seq\Delta, A \iimp B,\Lambda\}}{\Upsilon\{\Gamma\seq\Delta,\Lambda,\nes{A\seq B}\}}
\]
\[
\infer[{\land_L^n}]{\Upsilon\{\Gamma, A \land B \seq \Delta,\Lambda\}}{\Upsilon\{\Gamma,  A,B \seq \Delta,\Lambda\} }
\quad
\infer[{\land_R^n}]{\Upsilon\{\Gamma \seq A \land B, \Delta,\Lambda\}}{\Upsilon\{\Gamma \seq
 A, \Delta,\Lambda\} \quad \Upsilon\{\Gamma \seq B, \Delta,\Lambda\}}
\]
\[
\infer[{\lor_L^n}]{\Upsilon\{\Gamma, A \lor B \seq  \Delta,\Lambda\}}{\Upsilon\{\Gamma, 
A, \seq \Delta,\Lambda\} \quad \Upsilon\{\Gamma, B \seq  \Delta,\Lambda\}}
\quad
\infer[{\lor_R^n}]{\Upsilon\{\Gamma \seq A \lor B, \Delta,\Lambda\}}{\Upsilon\{\Gamma \seq A,B, \Delta,\Lambda\} }
\]
\[
\infer[\lift^n]{\Upsilon\{\Gamma, A\seq\Delta,\Lambda,\nes{\Gamma'\seq\Delta',\Lambda'}\}}{\Upsilon\{\Gamma,A\seq\Delta, \Lambda\nes{\Gamma', A\seq\Delta',\Lambda'}\}}
\quad
\infer[{\bot_L^n}]{\Upsilon\{\Gamma,\bot \seq \Delta,\Lambda\}}{} 
\quad
 \infer[{\init^n}]{\Upsilon\{\Gamma,P \seq P, \Delta,\Lambda\}}{} 
\]
\caption{Nested system $\NS_{\mLJ}$.}\label{fig:nestedmLJ}
\end{figure}

\subsubsection{Normal forms in $\NS$}\label{sec:nfnested}

While adding a tree structure to sequents enhances the expressiveness of the nesting framework when compared with the sequent one, the price to pay
is that
the obvious proof search procedure is of suboptimal complexity, since
it constructs potentially exponentially large nested
sequents~\cite{Brunnler:2009kx}. Hence the quest for proposing {\em proof search strategies} for  taming the proof search space (see e.g.~\cite{DBLP:conf/fossacs/ChaudhuriMS16}  for proof strategies  based on {\em focusing}). In what follows, we will propose a {\em normalisation procedure}  (hence also a proof strategy) that allows for reconciling the added extra superior expressiveness and
modularity of nested sequents over ordinary sequents with the computational behaviour of the
standard sequent framework.

Permutability of rules is often used in sequent systems in order to establish a notion of normal forms to proofs, where application of rules can be re-organised so to follow a determinate shape. 
\begin{definition}\label{def:perm}
Let $S$ be a nested sequent in a $\NS$. 
We say that a rule {\em $r_2$ permutes down $r_1$}
 ($r_2 \downarrow r_1$) if, for every derivation  in which 
  $r_1$ operates on $S$ and 
  $r_2$ operates on one
  or more of $r_1$'s premises (but not on auxiliary structures of
  $r_1$), there exists another derivation of $S$  in
  which $r_2$ operates on $S$ and $r_1$ operates on zero or more of
  $r_2$'s premises (but not on auxiliary structures of $r_2$).
 If $r_2\downarrow r_1$ and $r_1 \downarrow r_2$ we will say that $r_1,r_2$ are {\em permutable}. If all 
pair of rules in $\NS$ permute we say that $\NS$ {\em is fully permutable}.
\end{definition}

In the case of nested systems, permutability alone is often not enough for proposing 
effective proof strategies, since rules can be applied deep inside a nesting and also among nestings. 
In what follows, we will restrict the rules so to allow only exchange of formulae  in inner nestings.
\begin{definition}
Let 
$r$ be a  rule in $\NS$ with principal formula $A$ and auxiliary formulae $A_1,\ldots,A_k$. 
We say that $r$ is  {\em inter-nested} if  $A\prec A_i$ or $A\parallel A_i$ for some $i\in\{1,\ldots,k\}$.
$\NS$ is 
{\em n-directed} if, for  every 
inter-nested
rule $r$, $A\prec A_i,\forall i\in\{1,\ldots,k\}$.  $\NS$ is {\em shallow n-directed} if, for each $i\in\{1,\ldots,k\}$, there exists a $j_i\in\mathbb{N}$ such that $\id_{\Upsilon_i}(A_i)=\id_{\Upsilon}(A).j_i$.
\end{definition}
Observe that the nested rules for the modalities $\5$ and $\mathsf{b}$ {\em are not} n-directed (see~\cite{DBLP:conf/fossacs/ChaudhuriMS16}). It is interesting to note that there is no cut-free sequent calculus for $\Sfi$ either.
 
\begin{lemma}\label{lemma:prov}
Let $\NS$ be a 
n-directed nested system and let
$\Lambda_i$ be nestings in $\NS$. 
If
$\seq\Lambda_1,\ldots,\Lambda_n$ is derivable  then 
$\seq\Lambda_i$  is derivable for some $i\in\{1,\ldots,n\}$.
\end{lemma}
\begin{proof}
Let $\pi$ be a 
proof of $\seq\Lambda_1,\Lambda_2,\ldots,\Lambda_n$. 
Since $\NS$ is n-directed, any instance of a rule $r$ acting in $\Lambda_i$ 
is such that the auxiliary formulae remain in $\Lambda_i$. This means that they cannot be principal in any rule $s$ acting in $\Lambda_j, i\not=j$. 
Hence  $\pi$
can be re-ordered so that there is no interleaving of rules applied to different nestings. 

We observe that the theorem is valid in the other direction if weakening on nestings is available (\eg with a weakening absorbing initial axiom).\qed
\end{proof}
That is, in n-directed systems the comma on the right context is {\em intuitionistic}.
Moreover, if the n-directness is shallow, a depth first normalization procedure can be defined, which induces a {\em depth first proof strategy}. 
\begin{theorem}\label{thm:norm}
Let $\NS$ be a fully permutable, shallow n-directed nested system. 
Then any proof $\pi$ of a nested sequent $\Upsilon$ can be re-organised so that to satisfy the following 
normalisation procedure:
\begin{itemize}
\item {\em local phase:} apply, in any order, 
rules that do not move formulae between nestings;
\item {\em nesting phase:} apply, in any order, all possible  rules creating nestings; 
\item
{\em lift phase:} apply inter-nested rules; 
\item
{\em deep phase:} start the process again deep inside a sub-nesting;
\item Axioms are applied eagerly. 
\end{itemize}
\end{theorem}
\begin{proof} 
Let $\pi$ be a proof of $\Gamma\seq\Delta,\Lambda$ in $\NS$. Since $\NS$ is fully permutable, $\pi$ can be re-organised so to apply first all local rules acting in $\Gamma,\Delta$ then to create nestings via rules in the nesting phase. At this point, leaves in the (open) derivation will have the form $\Gamma'\seq\Delta',\Lambda'$
with $\Lambda\subseteq\Lambda'$. Since $\NS$ is shallow n-directed, rules in the nesting phase cannot add formulae to the sequent contexts 
and
rules in the lift phase have principal formulae in $\Gamma',\Delta'$ and auxiliary formulae in $\Lambda'$. Hence no more rules in $\pi$ are applied over formulae in the sequent contexts and the proof must continue in the nestings with no interleaving 
of rules applied to different nestings (Lemma~\ref{lemma:prov}).\qed
\end{proof}
We will show in Sec.~\ref{sec:int}, \ref{sec:mm}, \ref{sec:nn} representative examples of  systems falling into the class of 
fully permutable, shallow n-directed nested systems. 
A quick remark is that fully permutability can be substituted by a weaker condition: local rules permute down the nesting ones. 
 
\section{Linear nested systems}\label{sec:lnested}
In~\cite{DBLP:conf/lpar/LellmannOP17,DBLP:journals/corr/LellmannP17} we  proposed the concept of end-active, blocked form
 \emph{linear nested sequents}~\cite{Lellmann:2015lns}. 
In this section, we will show how such systems can both: be generated from 
shallow n-directed nested systems; and  recover
 sequent systems.
\begin{definition}
The set $LNS$ of \emph{linear nested sequents} is given recursively by:
\begin{enumerate}
  \item if  $\Gamma\seq \Delta$ is a sequent then $\Gamma\seq \Delta\in LNS$
  \item if $\Gamma\seq \Delta$ is a sequent and $\mathcal{G}\in LNS$
    then $\Gamma\seq\Delta\lns\mathcal{G}\in LNS$.
\end{enumerate}
A {\em linear nested system  ($\LNS$)} consists of a set of inference rules acting on
linear nested sequents.
We call each sequent
in a linear nested sequent a {\em component}.

An application of a linear nested sequent rule is \emph{end-active}
if the rightmost components of the premises are active and the only
active components (in premise and conclusion) are the two rightmost
ones.
\end{definition}
In other words, a linear nested sequent 
is  a finite list of sequents. Since this data structure matches   a path in a nested tree,
Lemma~\ref{lemma:prov}  immediately entails that proofs in n-directed nested systems
can be linearised. 
This linearisation process can also be extended to nested rules. We will illustrate 
this formalisation for the shallow case.

Let $\NS$ be shallow n-directed. Observe that the depth first procedure assures that nested rules can be restricted so to be applied in the two deepest sub-nestings of a nesting. This implies 
that the nested sequents in conclusion and premises  in any rule have the form $\Upsilon\{\Gamma\seq\Delta,\Upsilon_1\{\,\},\ldots,\Upsilon_n\{\,\}\}$ with $\dpt{\Upsilon_n\{\,\}}\leq 1$ (see Def.~\ref{def:holecont}). After the local phase, rules do not alter the sequent contexts $\Gamma,\Delta$. Thus, after the lift phase, a sequent $\Gamma\seq\Delta,\Lambda_1,\ldots,\Lambda_n$ will be provable iff
$\seq\Lambda_1,\ldots,\Lambda_n$ is provable.
In view of Lemma~\ref{lemma:prov}, this holds iff $\seq\Lambda_i$ is provable, for some $i\in\{1,\ldots,n\}$. Hence the nesting phase can be restricted to the creation of a single nesting and nested contexts of the form $\Upsilon\{\Gamma\seq\Delta,\nes{\Gamma'\seq\Delta'}\}$ present on inference rules can be written as
the linear nested sequent 
$\Gscr\lns \Gamma\seq\Delta\lns\Gamma'\seq\Delta'$, where $\Gscr$ carries the 
position of the sequent $\Gamma'\seq\Delta'$ in $\Upsilon\{\,\}$.

Fig.~\ref{fig:lnsint} presents  the end-active system  $\LNS_{\mLJ}$~\cite{Lellmann:2015lns}, which is the linearisation of the system $\NS_{\mLJ}$.

\begin{figure}[t]
\[
\infer[{\iimp_L^l}]{\Gscr\lns\Gamma, A\iimp B \seq \Delta}{\Gscr\lns\Gamma,  A\iimp B
 \seq A, \Delta \quad \Gscr\lns\Gamma, B \seq
\Delta }
\quad
\infer[{\iimp_R^l}]{\Gscr\lns\Gamma \seq A \iimp B, \Delta}{\Gscr\lns\Gamma\seq\Delta\lns A
\seq B}
\quad
\infer[{\land_L^l}]{\Gscr\lns\Gamma, A \land B \seq \Delta}{\Gscr\lns\Gamma,  A,B \seq \Delta }
\]
\[
\infer[{\land_R^l}]{\Gscr\lns\Gamma \seq A \land B, \Delta}{\Gscr\lns\Gamma \seq
 A, \Delta \quad \Gamma \seq B, \Delta}
\quad
\infer[{\lor_L^l}]{\Gscr\lns\Gamma, A \lor B \seq  \Delta}{\Gscr\lns\Gamma, 
A, \seq \Delta \quad \Gscr\lns\Gamma, B \seq  \Delta}
\quad
\infer[{\lor_R^l}]{\Gscr\lns\Gamma \seq A \lor B, \Delta}{\Gscr\lns\Gamma \seq A,B, \Delta }
\]
\[
 \infer[{\lift^l}]{\Gscr\lns\Gamma,A \seq \Delta\lns \Gamma'\seq\Delta'}{\Gscr\lns\Gamma \seq \Delta\lns A,\Gamma'\seq\Delta'} 
\quad
\infer[{\bot_L^l}]{\Gscr\lns\Gamma,\bot \seq \Delta}{} 
 \quad
 \infer[{\init^l}]{\Gscr\lns\Gamma,A \seq A, \Delta}{} 
\]
\caption{End-active linear nested system $\LNS_{\mLJ}$.}\label{fig:lnsint}
\end{figure}

\subsubsection{Recovering sequent systems}
Observe that 
end-activeness in linear nested systems
alone is not enough for 
faithfully translating the depth first normalisation strategy  since the local  phase
can be interleaved with the lift and deep phases. For ensuring such strategy, we  add the auxiliary nesting operator $\lnse$\,  to the premises of rules in the nesting phase,
which is switched to the the original nesting operator $\lns$\, after the application of rules in the lift phase. This assures that rules in nesting and lift phases are restricted so that to  occur  in a block~\cite{DBLP:conf/lpar/LellmannOP17,DBLP:journals/corr/LellmannP17}. 

This procedure determine that the rules applied in the nesting + lift phases can be seen as a single {\em macro rule},  introducing a {\em synthetic connective}~\cite{DBLP:conf/fossacs/ChaudhuriMS16}.
That means that, modulo the order of instances of applications of  rules in the lift phase, there is a 1-1 correspondence between
derivations in the blocked form,  end-active variant of  $\LNS$   and derivations in  
the derived $\SC$.

In Fig.~\ref{fig:blocked}
we present a {\em blocked form} system for $\LNS_{\mLJ}$. Observe that
applying $\iimp_R^b$ then $\lift^b$ (so that to transfer all the formulae in the left context to the last component) then $\cls^b$ gives the rule $\iimp_R$ in Fig.~\ref{fig:mLJ}.
\begin{figure}[t]
$\qquad\qquad
\infer[{\iimp_R^b}]{\Gscr\lns\Gamma \seq A \iimp B, \Delta}{\Gscr\lns\Gamma\seq\Delta\lnse A
\seq B}
\qquad
 \infer[{\lift^b}]{\Gscr\lns\Gamma,A \seq \Delta\lnse \Gamma'\seq\Delta'}{\Gscr\lns\Gamma \seq \Delta\lnse A,\Gamma'\seq\Delta'} 
 \qquad
 \infer[\cls^b]{\mathcal{G} \lnse \Gamma\seq \Delta}{\mathcal{G}
    \lns \Gamma\seq \Delta }
$
\caption{Blocked form version for $\LNS_{\mLJ}$.}\label{fig:blocked}
\end{figure}

\section{Labelled proof systems}
\label{sec:lbns}

While it is widely accepted that nested systems carry  the Kripke structure on nestings for intuitionistic and normal modal logics, it is not clear what is the relationship between nestings and semantics for other systems. For example, in~\cite{DBLP:conf/lpar/LellmannOP17} we presented a $\LNS$ for linear logic ($\LL$), but the interpretation of nestings for this case is still an open problem.

In this work we will relate labelled nested systems~\cite{DBLP:conf/aiml/GoreR12}
with labelled systems~\cite{DBLP:books/daglib/0003059}. While the results for 
intuitionistic and modal logics are not new~\cite{Fitting:2014,Negri:2005p878}, we present the first semantical interpretation for nestings in non-normal modal logics. In this section we shall recall some of the terminology for labelled systems.
\subsubsection{Labelled nested systems}\label{subsec:lbns}
Let $\sv$ a countable infinite set of {\em state variables} (denoted by $x, y, z,
\ldots$), disjoint from the set of propositional variables. A {\em labelled formula} has the form $x : A$ where $x\in\sv $ and $A$ is a formula. 
If $\Gamma = \{A_1,\ldots, A_n\}$ is a multiset of formulae, then $x :\Gamma$ denotes the multiset $\{x : A_1,\ldots , x : A_n\}$ of labelled formulae. 
A (possibly empty) set of relation terms (\ie\ terms
of the form $xRy$, where $x,y\in\sv$) is called a {\em relation set}. For a relation set
$\Rscr$, the \emph{frame} $Fr(\Rscr)$ defined by $\Rscr$ is given by
$(|\Rscr|,\Rscr)$ where $|\Rscr| = \{x\; |\; xRy \in \Rscr \mbox{ or }
yRx \in \Rscr \mbox{ for some }y\in\sv\}$.  We say that a relation set
$\Rscr$ is {\em treelike} if the frame defined by $\Rscr$ is a tree or
$\Rscr$ is empty. A treelike relation set $\Rscr$ is called {\em
  linelike} if each node in $\Rscr$ has at most one child.
\begin{definition}\label{def:lns} 
  A {\em labelled nested sequent} $\LLNS$ is a labelled sequent
  $\Rscr,X\seq Y$ where
\begin{enumerate}
\item $\Rscr$ is treelike;
\item if $\Rscr=\emptyset$ then $X$ has the form $x\lab
  A_1,\ldots,x\lab A_n$ and $Y$ has the form $x \lab B_1,\ldots, x
  \lab B_m$ for some  $x\in\sv$;
\item if $\Rscr\not =\emptyset$ then every state variable $y$ that
  occurs in either $X$ or $Y$ also occurs in $\Rscr$.
\end{enumerate}
  A {\em labelled nested sequent calculus} 
is a labelled 
calculus whose initial sequents and inference rules are constructed
  from $\LLNS$. 
\end{definition}

 As in~\cite{DBLP:conf/aiml/GoreR12} labelled nested systems can be automatically generated from nested systems. 
\begin{definition}\label{def:Tx}
Given $\Gamma\seq\Delta$ and $\Gamma'\seq\Delta'$ sequents, 
we define $(\Gamma\seq\Delta)\tensor(\Gamma'\seq\Delta')$
to be $\Gamma,\Gamma'\seq\Delta,\Delta'$.
For a state variable $x$, define the mapping $\mathbb{TL}_{x}$ from
 $\NS$ to $\LLS$
as follows
$$
\begin{array}{lcl}
\mathbb{TL}_{x}(\Gamma \seq \Delta,\nes{\Upsilon_1},\ldots,\nes{\Upsilon_n}) &=&
xRx_1,\ldots, xRx_n, 
( x\lab \Gamma \seq x\lab \Delta)\;\tensor\\
& &\mathbb{TL}_{x_1}(\Upsilon_1)\tensor\ldots \tensor
\mathbb{TL}_{x_n}(\Upsilon_n)\\
\mathbb{TL}_{x}(\nes{\Gamma \seq \Delta})&=&x\lab \Gamma \seq x\lab \Delta\\
\end{array}
$$
with all state variables pairwise distinct.
\end{definition}
For the sake of readability, when the state variable is not important, we will suppress the subscript, writing $\mathbb{TL}$ instead of $\mathbb{TL}_{x}$ .
We will shortly illustrate the procedure of constructing labelled nestings using the mapping $\mathbb{TL}$.
Consider the following application of the rule $\iimp_R$ of 
Fig.~\ref{fig:nestedmLJ}:
$$
\infer[\iimp_R^n]{\Upsilon\{\Gamma\seq\Delta, A \iimp B,\Lambda\}}{\Upsilon\{\Gamma\seq\Delta,\Lambda,\nes{A\seq B}\}}
$$
Applying $\mathbb{TL}$ to the conclusion we obtain 
$\Rscr,X\seq Y,x\lab A\iimp B$ where the variable $x$ label 
formulae in two components of the $\NS$, and $X,Y$ are multisets of labelled formulae.
Applying $\mathbb{TL}$ to the
premise we obtain 
$\Rscr,xRy,X,y\lab A\seq Y,y\lab B$
where $y$ is a fresh variable (\ie\ different from $x$ and not occurring in $X,Y$).
We thus obtain an application of the $\LLS$ rule
\[\infer[\mathbb{TL}(\iimp_R^n)]{\Rscr,X\seq Y,x\lab A\iimp B}{\Rscr,xRy,X,y\lab A\seq Y,y\lab B}\]
Some rules of the labelled nested system $\LLNS_\mLJ$ are depicted in Fig.~\ref{fig:lab-nes-mLJ}.

\begin{figure}[t]
\[\infer[{\mathbb{TL}({\iimp_L^n})}]{\Rscr,X, x\lab A\iimp B \seq Y}{\Rscr,X,
 \seq x\lab A, Y \quad X, x\lab B \seq
Y }
\qquad
\infer[\mathbb{TL}({\land_L^n})]{\Rscr,X, x\lab A \land B \seq Y}{\Rscr,X,  x\lab A,x\lab B \seq Y }
\]
\[
\infer[\mathbb{TL}({\land_R^n})]{\Rscr,X \seq x\lab A \land B, Y}{\Rscr,X \seq
 x\lab A, Y \quad X \seq x\lab B, Y}
\qquad
\infer[\mathbb{TL}({\lor_L^n})]{\Rscr,X, x\lab A \lor B \seq  Y}{\Rscr,X, 
x\lab A, \seq Y \quad X, x\lab B \seq  Y}
\]
\[
\infer[\mathbb{TL}({\lor_R^n})]{\Rscr,X \seq x\lab A \lor x\lab B, Y}{\Rscr,X \seq x\lab A,x\lab B, Y }
\qquad
 \infer[{\mathbb{TL}(\init^n)}]{\Rscr,X,x\lab A \seq x\lab A, Y}{} 
\]
\[
\infer[\mathbb{TL}({\bot_L^n})]{\Rscr,X,x\lab \bot \seq Y}{} 
\quad
\infer[\mathbb{TL}(\iimp_R^n)]{\Rscr,X\seq Y,x\lab A\iimp B}{\Rscr,xRy,X,y\lab A\seq Y,y\lab B}
\quad
\infer[{\mathbb{TL}(\lift^n)}]{\Rscr,xRy,X, x\lab A\seq Y}{\Rscr,xRy,X, y\lab A\seq Y}
\]
\caption{Labelled nested system $\LLNS_{\mLJ}$.}\label{fig:lab-nes-mLJ}
\end{figure}
The following result follows readily by transforming derivations bottom-up~\cite{DBLP:conf/aiml/GoreR12}.
\begin{theorem}\label{trans}
The mapping $\mathbb{TL}_{x}$ preserves open derivations, that is, there is a 1-1 correspondence between derivations in 
a nested sequent system $\NS$ and in its labelled translation $\LLNS$.
\end{theorem}


\subsubsection{Labelled sequent systems}\label{sec:labelsec}
In the labelled sequent framework, a semantical characterisation of a logic  is explicitly added to sequents via the labelling of formulae~\cite{DBLP:journals/jphil/Mints97,DBLP:books/daglib/0003059,DBLP:journals/aml/DyckhoffN12,Negri:2005p878,Negri2017}.  
In the case of world based semantics, the forcing relation $x\Vdash A$  is represented as the labelled formula $x \lab A$ and sequents have the form $\Rscr,X\seq Y$, where $\mathcal{R}$ is a relation set and $X,Y$ are multisets of labelled
  formulae.

The rules of the labelled calculus $\GtI$  are obtained from the inductive definition of validity in a Kripke frame (Fig.~\ref{fig:LS-mLJ}), together with the rules describing a partial order, 
presented in Fig.~\ref{fig:relational}. 
Note that the anti-symmetry rule does not need to be stated directly since, for any $x$, the formula $x=x$ is equivalent to true and hence can be erased from the left side of a sequent.

\begin{figure}[t]
\begin{subfigure}{.65\linewidth}
$
\begin{array}{c}
\infer[{\iimp_L^t}]{\Rscr,x\leq y, X, x\lab A\iimp B \seq Y}{\Rscr,x\leq y, X, x\lab A\iimp B
 \seq y\lab A, Y \quad \Rscr,x\leq y, X, y\lab B \seq
Y }
\\\\
\infer[\iimp_R^t]{\Rscr,X\seq Y,x\lab A\iimp B}{\Rscr,x\leq y,X,y\lab A\seq Y,y\lab B}\quad
\infer[\init^t]{\Rscr,X,x\leq y, x\lab P \seq Y,y\lab P}{} 
\end{array}
$
\caption{$y$ is fresh in  $\iimp_R$  and $P$ is atomic in $\init$.}\label{fig:LS-mLJ}
\end{subfigure}
\begin{subfigure}{.3\linewidth}
$\begin{array}{cc}
\infer[\text{Ref}]{\Rscr,X \seq Y}{xRx, \Rscr,X \seq Y} \\\\
  \infer[\text{Trans}]{xRy, yRz, \Rscr,X \seq Y}{xRz, xRy, yRz, \Rscr,X \seq Y}  \end{array}
$
\caption{Relation rules.}\label{fig:relational}
\end{subfigure}
\caption{Some rules of the labelled system $\GtI$}
\end{figure}

\section{Intuitionistic logic}\label{sec:int}

In this section we will give a tour on various proof systems for intuitionistic logic, relating them by applying the results presented in the last sections.

\begin{theorem}\label{prop:inv-perm}
Weakening is height-preserving admissible in $\NS_{\mLJ}$.
Moreover, all introduction rules in $\NS_{\mLJ}$ are 
invertible where both the height of the derivation and the minimal level of the active components of rule
applications are preserved. 
Finally,  
$\NS_{\mLJ}$ is fully permutable.
\end{theorem}
\begin{proof}
The proofs of weakening-admissibility and invertibility are by induction on the depth of the derivation, distinguishing cases according to the last applied rule.
Permutability of rules is proven by a case-by-case analysis and it uses the 
invertibility and weakening results. \qed
\end{proof}
\begin{remark}
Observe that permutability also holds for the $\lift$ rule. In fact the case
$$
\infer[\iimp_R]{A,\Gamma\seq,\Delta,\Lambda,B\iimp C}
{\infer[\lift]{A,\Gamma\seq,\Delta,\Lambda,\nes{B\seq C}}
{A,\Gamma\seq,\Delta,\Lambda,\nes{A,B\seq C}}}
$$
is not considered for permutation since $\nes{B\seq C}$ is an auxiliary structure 
of the principal formula $B\iimp C$ (see Def.~\ref{def:perm}).
\end{remark}
Since $\NS_{\mLJ}$ is shallow n-directed,  the depth first strategy holds 
with the following classification of rules: 
 {\em local phase:} conjunction, disjunction 
and implication left
rules;
 {\em nesting phase:} implication right rule;
{\em lift phase:} $\lift$ rule.

The results in the previous sections entail the following.
\begin{theorem}\label{prop:equiv1}
Systems $\NS_{\mLJ}$, $\LNS_{\mLJ}$, $\mLJ$ and $\LLNS_{\mLJ}$  are equivalent.
\end{theorem}
Observe that the proof uses syntactical arguments only, differently from \eg~\cite{Fitting:2014,Lellmann:2015lns}. 

For establishing a comparison between labels in $\GtI$ and $\LLNS_{\mLJ}$, first observe that
 applications of rule Trans in $\GtI$ can be restricted to the leaves (\ie\ just before an instance of the initial axiom).
Also, since weakening is admissible in $\GtI$ and 
 the monotonicity property:
$x \Vdash A\mbox{ and } x\leq y\mbox{ implies }y\Vdash A$ is derivable in $\GtI$ (Lemma 4.1 in~\cite{DBLP:journals/aml/DyckhoffN12}), the next result follows.
\begin{lemma}
The following  
rules are {\em derivable} in $\GtI$ up to weakening.
$$
\infer[{\iimp_{L'}}]{\Rscr, X, x\lab A\iimp B \seq Y}{\Rscr, X, x\lab A\iimp B
 \seq x\lab A, Y \quad \Rscr, X, x\lab B \seq
Y }
\qquad
\infer[\init']{\Rscr,X,x\lab P \seq Y,x\lab P}{} 
$$
Moreover, the rule
$$
\infer[\lift']{\Rscr,x\leq y,X, x\lab A\seq Y}{\Rscr,x\leq y,X, y\lab A\seq Y}
$$
is admissible in $\GtI$. 
\end{lemma}
\begin{proof}
For the derivable rules, just note that
$$
\infer[\text{Ref}]{\Rscr, X, x\lab A\iimp B \seq Y}
{\infer[\iimp_L^t]{\Rscr, X,x\leq x, x\lab A\iimp B \seq Y}
{\deduce{\Rscr, X,x\leq x, x\lab A\iimp B  \seq Y,x\lab A}{}&
\deduce{\Rscr, X,x\leq x, x\lab B \seq Y}{}}}
\quad
\infer[\text{Ref}]{\Rscr, X, x\lab P, \seq Y,x\lab P}
{\infer[\init^t]{\Rscr,x\leq x,X,x\lab P \seq Y,x\lab P}{}} 
$$\qed
\end{proof}

Using an argument similar to the one in~\cite{DBLP:conf/aiml/GoreR12}, it is easy to see that, in the presence of the primed rules shown above,
 the relational rules are admissible. Moreover, labels are preserved.
\begin{theorem}
\label{theo:GtI-LLNS}
$\GtI$  is  label-preserving equivalent to $\LLNS_{\mLJ}$.
\end{theorem}
That is, nestings in $\NS_{\mLJ}$ and $\LNS_{\mLJ}$ correspond to worlds in 
the Kripke structure where the sequent is valid and this is the semantical interpretation of the linear nested system for intuitionistic logic~\cite{Fitting:2014}. 
 
\subsubsection{Wrapping up}
Observe that, in labelled line systems, the relation $R$ relates two components in a sequence, hence $\leq$ is the 
transitive
closure of $R$ in any nested path. 
Since derivability is the same in $\LNS_{\mLJ}$ and $\LLS_{\mLJ}$ (Thm.~\ref{trans}),
this means that
a proof in $\LNS_{\mLJ}$ corresponds to a successful trace in a deep-first proof strategy in $\GtI$, which, by its turn, corresponds to 
a path in the Kripke semantics of intuitionistic logic.

Moreover, more than internalising the semantics, end-active linear nested systems actually show that we may consider only the upper most words in the Kripke semantics. 

Finally, since $\mLJ$ derivations are equivalent to blocked end-active $\LNS_{\mLJ}$
derivations, the semantical analysis for $\LNS_{\mLJ}$ also hold for $\mLJ$.  

In what follows, we will show how this tour on different proof systems can be smoothly extended to (multi-) normal modalities and to non-normal modalities, using propositional classical logic as the base logic. 
 
\section{Multi-modal logics}\label{sec:mm}


In~\cite{DBLP:conf/lpar/LellmannP15} we presented sequent calculi and linear nested systems for  multimodal logics with a simple
interaction between  modalities, called \emph{simply dependent
multimodal logics}~\cite{Demri:2000}. 
The language for these logics
contains indexed modalities $\ibox{i}$ for indices $i$ from an index
set $N \subseteq \nat$ of natural numbers. The axioms are given by
extensions of the axioms of modal logic $\K$ for every modality
$\ibox{i}$ together with \emph{interaction axioms} of the form
$\ibox{i} A \iimp \ibox{j} A$.  

In this paper we will consider 
a subset of these logics such that, for each index, the underline logic is
an extension of
$\K$  with axioms from the set $\{\D,\T,\4\}$
(see Fig.~\ref{fig:modal-axioms}). 

\begin{figure}[t]
$    \K\; \Box (A \iimp B) \iimp (\Box A \iimp \Box B)
    \qquad
    \vcenter{\infer[\mathsf{nec}]{\Box A}{A}}
    \qquad
    \D\; \neg \Box \bot
    \qquad
    \T\; \Box A \iimp A
    \qquad
    \4\; \Box A \iimp \Box \Box A
$  \caption{Modal logic $\K$ contains
    the propositional tautologies, modus ponens, $\K$ and
    $\mathsf{nec}$.}\label{fig:modal-axioms}
\end{figure}

\begin{definition}\label{def:sdml}
A  {\em simply dependent multimodal logical system}
is given by a triple $(N,\cless,F)$, where $N$ is a finite set of
natural numbers, $(N,\cless)$ is a partial order and $F$ is a mapping from $N$ to  $2^{\{ \mathsf{D}, \mathsf{T}, \mathsf{4}\}}$. Moreover,  $(N,\cless,F)$ is \emph{transitive-closed}, that is, for
  every $i,j\in N, j\cless i$, if $\K\4 \subseteq F(j)$ then
  $\K\4\subseteq F(i)$.

 The \emph{logic described by}
$(N,\cless,F)$  has modalities $\ibox{i}$ for every $i \in N$  with  
the axioms and rules of classical propositional logic
together
with rules and
axioms for the
modality $i$ given by the necessitation rule and the $\K$
axiom for $\ibox{i}$ as well as the axioms $F(i)$, and interaction axioms
$\ibox{i} A\iimp \ibox{j} A$ for every $i,j \in I$ with $j\cless
i$, understood as zero-premise rules.
\end{definition}
\begin{remark} Observe that  $\K$ and its $\{\D,\T,\4\}$-extensions are 
trivial cases of simply dependent multimodal logics where the index set
$N$ is a singleton. This means that all the results stated in what follows hold for, \eg, $\Sf$ and $\K\D$. Also note that the modal axioms $\D$ and $\T$ propagate upwards, so there is no need for reflexive or serial closure.
\end{remark}

The following definition extends the concept of frames to simply
dependent multimodal logic.  The notions of valuations, model and
truth in a world of the model are defined as usual (see~\cite{Blackburn:2001fk}).
\begin{definition}
  Let $(N,\cless,F)$ be a description for a simply dependent
  multimodal logic. A {\em $(N,\cless,F)$-frame} is a tuple
  $(W,(R_i)_{i\in N})$ consisting of a set $W$ of \emph{worlds} and an
  \emph{accessibility relation} $R_i$ for every index $i \in N$, such
  that for all $i,j \in N$:
  \begin{itemize}
  \item If the logic $F(i)$ contains $\K\D$, then $R_i$ is serial.
  \item If the logic $F(i)$ contains $\K\T$, then $R_i$ is reflexive.
  \item If the logic $F(i)$ contains $\K\4$, then $R_i$ is transitive.
  \item If $j \cless i$, then $R_j \subseteq R_i$.
  \end{itemize}
\end{definition}
Since here we only consider simply dependent multimodal logics where
the different component logics are extensions of $\K$ with axioms from
$\{\D,\T,\4\}$, and since the interaction axioms are of a particularly
simple shape, 
standard results e.g.\ from Sahlqvist
theory immediately yield
soundness and completeness of the description $(N,\cless,F)$ w.r.t. the logic
of the class of $(N,\cless,F)$-frames.

\subsection{Indexed nested systems}
Nestings can  be extended to multi-modalities by adding indexes.
An indexed nested sequent~\cite{Fitting2015}, is a nested sequent where each sequent node  is marked with an index taken from $N$, and it is denoted by
$\Gamma\seq\Delta,[G_1]^{i_1},[G_2]^{i_2},\ldots,[G_n]^{i_n}$.

We will denote by $\upset{j}$  the \emph{upset} of
the index $j$, i.e., the set $\{i \in N : j \cless i\}$ and extend this notation to the sets
$\upset[\mathsf{Ax}]{j} \defs \{ i \in N : j \cless i, \,
\K\mathsf{Ax}\subseteq F(i)\}$ and
$\upset[\neg\mathsf{Ax}]{j} \defs \{ i \in N : j \cless i, \,
\K\mathsf{Ax}\not\subseteq F(i)\}$ where $\mathsf{Ax}$ is any of the
axioms $\D,\T,\4$. Thus e.g.\ the set $\upset[\neg\4]{j}$ is the set
of indices $i$ with $j\cless i$ such that $\K\4 \not\subseteq F(i)$,
i.e., the logic $F(i)$ does not derive the transitivity axiom $\4$.

$\NS_{(N,\cless,F)}$ (see Fig.~\ref{fig:nestmulti}) is the nested sequent system for the  $(N,\cless,F)$ description for a simply dependent multimodal
logic.
\begin{figure}[t]
$
\begin{array}{c}
  \infer[{\ibox{ij}}_L^n]{\Upsilon\{\Gamma, \Box_iA \seq \Delta,
    \nes{\Sigma \seq \Pi}^j,\Lambda\}}{\Upsilon\{\Gamma, \Box_iA \seq \Delta,
     \nes{\Sigma,A \seq \Pi}^j,\Lambda\}}
  \qquad
  \infer[{\ibox{i}}_R^n]{\Upsilon\{  \Gamma \seq \Delta, \ibox{i}
    A,\Lambda\}}{\Upsilon\{\Gamma \seq \Delta, \nes{\;\seq A}^i,\Lambda\}}
\\\\
  \infer[\mathsf{d}_{ij}^n]{\Upsilon\{ \Gamma, \ibox{i} A \seq
    \Delta,\Lambda\}}{\Upsilon\{ \Gamma, \ibox{i} A \seq \Delta, \nes{A \seq \;}^j,\Lambda\}}
  \qquad
  \infer[\mathsf{t}_i^n]{\Upsilon\{ \Gamma, \ibox{i} A \seq
    \Delta,\Lambda\}}{\Upsilon\{ \Gamma, \ibox{i} A, A \seq \Delta,\Lambda\}}
  \qquad
  \infer[{\4_{ij}^n}]{\Upsilon\{\Gamma, \ibox{i} A \seq \Delta,
    \nes{\Sigma \seq \Pi}^j,\Lambda\}}{\Upsilon\{\Gamma, \Box_iA \seq \Delta,
    \nes{\Sigma,\ibox{i}A \seq \Pi}^j,\Lambda\}}
\end{array}
$
  \begin{align*}
  \mathcal{S}_{(N,\cless, F)} \defs &\; \{ {\ibox{i}}_R^n : i \in N\} \cup \{
  {\ibox{ij}}_L^n : i,j \in N, i \in \upset{j}\} \cup \{ \mathsf{d}_{ij}^n :
  i,j \in N, i \in \upset[\D]{j}\}\\
  & 
   \cup \{ \mathsf{t}_i^n : i \in N, \K\T \subseteq F(i)\} 
  \cup \{
  {\4_{ij}^n} : i,j \in N, i\in \upset[\4]{j}\}
  \end{align*}
    \caption{Indexed nested sequent rules for $\NS_{(N,\cless,F)}$.}
  \label{fig:nestmulti}
\end{figure}
Next result follows the same lines as in Sec.~\ref{sec:nested}.
\begin{theorem} $\NS_{(N,\cless,F)}$ has the following properties:
 weakening is height-preserving admissible;
 all introduction rules in are height-preserving invertible;
any pair of rules  is permutable.
\end{theorem}
$\NS_{(N,\cless,F)}$ is shallow n-directed, hence the depth first strategy holds with the following classification of rules: {\em local phase}: propositional and $\mathsf{t}_i$ rules; {\em nesting phase}: $\mathsf{d}_i$ and $\ibox{i}_R$ rules; {\em lift phase}: ${\ibox{ij}}_L$ and ${\4{ij}}_L$ rules.
This allows the definition of a linear version of nested systems for  the simply dependent
    multimodal logic given by the description $(N,\cless,F)$, system $\LNS_{(N,\cless,F)}$ presented in Fig.~\ref{fig:lns-rules-simply-dep-mult}. 
    \begin{figure}[t]
$
\begin{array}{c}
   \infer[{\ibox{i}}_R^l]{\mathcal{G}\lns[k]  \Gamma \seq \Delta, \ibox{i}
    A}{\mathcal{G}\lns[k]  \Gamma \seq \Delta\, \lnse[i] \;\seq A}
    \qquad 
     \infer[{\ibox{ij}}_L^l]{\Gscr\lns[k] \Gamma, \ibox{i} A \seq \Delta
    \lnse[j]\Sigma \seq \Pi}{\Gscr\lns\Gamma,\ibox{i} A \seq \Delta
    \lnse[j]\Sigma,A \seq \Pi}
    \qquad
     \infer[\cls]{\mathcal{G} \lnse[k] \Gamma\seq \Delta}{\mathcal{G}
    \lns[k] \Gamma\seq \Delta }
\\\\
        \infer[\mathsf{d}_{ij}^l]{\mathcal{G}\lns[k] \Gamma, \ibox{i} A \seq
    \Delta}{\mathcal{G} \lns[k] \Gamma \seq \Delta \lnse[j] A \seq \;}
\qquad
   \infer[\mathsf{t}_i^l]{\Gscr\lns[k] \Gamma, \ibox{i} A \seq
    \Delta}{\Gscr\lns[k] \Gamma, \ibox{i} A, A \seq \Delta}
  \qquad
  \infer[{\4_{ij}^l}]{\Gscr\lns[k] \Gamma, \ibox{i} A \seq \Delta
    \lnse[j]\Sigma \seq \Pi}{\Gscr\lns[k]\Gamma,\ibox{i} A \seq \Delta
    \lnse[j]\Sigma,\ibox{i}A \seq \Pi}
\end{array}
$
  \begin{align*}
  \mathcal{S}_{(N,\cless, F)} \defs &\; \{ {\ibox{i}}_R : i \in N\} \cup \{
  {\ibox{ij}}_L : i,j \in N, i \in \upset{j}\} \cup \{ \mathsf{d}_{ij} :
  i,j \in N, i \in \upset[\D]{j}\}\\
  & 
   \cup \{ \mathsf{t}_i : i \in N, \K\T \subseteq F(i)\} 
  \cup \{
  {\4_{ij}} : i,j \in N, i\in \upset[\4]{j}\}
  \end{align*}
    \caption{End-active, blocked form system $\LNS_{(N,\cless,F)}$.}
  \label{fig:lns-rules-simply-dep-mult}
\end{figure}
    Observe that 
rules $\mathsf{d}_{ij}^l$ and $\ibox{i}_R^l$ are not invertible in $\LNS_{(N,\cless,F)}$.
Using blocked forms, the linear nested system emulates the rules of the sequent system $\SC_{(N,\cless,F)}$ in Fig.~\ref{fig:simply-dep-multimodal}~\cite{DBLP:journals/corr/LellmannP17}, as  illustrated next  for the rule $\mathsf{d}_i$:

\resizebox{\columnwidth}{!}{
$
\infer[{\mathsf{d}_{ij}}]{\Gscr\lns[k]\Omega,\{\ibox{i} \Gamma_i,\ibox{i} \Sigma_i
      : i \in \upset[\4]{j}\}, \{\{\ibox{i}A_i\} : i\in
      \upset[\neg\4]{j}\} \seq \Delta}
      {\infer=[{\ibox{ij}}_L]{\Gscr\lns[k]\Omega,\{\ibox{i} \Gamma_i,\ibox{i} \Sigma_i
      : i \in \upset[\4]{j}\}, \{\{\ibox{i}A_i\} : i\in
      \upset[\neg\4]{j}\} \seq \Delta\lnse[j] A_1\seq}
      {\infer=[\4_{ij}]{\Gscr\lns[k]\Omega,\{\ibox{i} \Gamma,\ibox{i} \Sigma_i
      : i \in \upset[\4]{j}\}, \{\{\ibox{i}A_i\} : i\in
      \upset[\neg\4]{j}\}  \seq\Delta\lnse[j]  \Sigma_i, A_1,\ldots,A_n\seq }
      {\infer[\cls]{\Gscr\lns[k]\Omega,\{\ibox{i} \Gamma,\ibox{i} \Sigma_i
      : i \in \upset[\4]{j}\}, \{\{\ibox{i}A_i\} : i\in
      \upset[\neg\4]{j}\} \seq\Delta\lnse[j] \{\ibox{i} \Gamma_i,\ibox{i} \Sigma_i
      : i \in \upset[\4]{j}\},  \Sigma_i, A_1,\ldots,A_n\seq }
      {\deduce{\Gscr\lns[k]\Omega,\{\ibox{i} \Gamma_i,\ibox{i} \Sigma_i
      : i \in \upset[\4]{j}\}, \{\{\ibox{i}A_i\} : i\in
      \upset[\neg\4]{j}\}  \seq\Delta\lns[j] \{\ibox{i} \Gamma_i,\ibox{i} \Sigma_i
      : i \in \upset[\4]{j}\},  \Sigma_i, A_1,\ldots,A_n\seq }{}}}}}
$}

\begin{figure}[t]
$
\begin{array}{c}
    \infer[\mathsf{k}_i]{\Omega,\{\ibox{j} \Gamma_j, \ibox{j} \Sigma_j
      : j\in \upset[\4]{i}\}, \{ \ibox{j} \Sigma_j :
      j\in\upset[\neg\4]{i}\} \seq \ibox{i} A, \Xi}{\{\ibox{j}
      \Gamma_j, \ibox{j} \Sigma_j, \Sigma_j : j \in \upset[\4]{i}\},
      \{ \Sigma_j : j \in \upset[\neg\4]{i}\} \seq A}
\\\\
    \infer[\mathsf{d}_i]{\Omega,\{\ibox{j} \Gamma_j, \ibox{j} \Sigma_j
      : j \in \upset[\4]{i}\}, \{ \ibox{j} \Sigma_j : j\in
      \upset[\neg\4]{i}\} \seq \Xi}{\{\ibox{j} \Gamma_j, \ibox{j}
      \Sigma_j, \Sigma_j : j \in \upset[\4]{i}\}, \{ \Sigma_j : j \in
      \upset[\neg\4]{i}\} \seq \;} \qquad\quad
    \infer[\mathsf{t}_i]{\Omega, \{\ibox{j} \Sigma_j : j \in
      \upset{i}\} \seq \Xi }{\Omega, \{\ibox{j} \Sigma_j,\Sigma_j : j
      \in \upset{i} \} \seq \Xi}
\end{array}
$
  \[
    \mathcal{S}_{(N,\cless,F)} \defs \{ \mathsf{k}_i : i \in N \} \cup
    \{ \mathsf{d}_i : i \in N, \K\D \subseteq F(i)\} \cup \{
    \mathsf{t}_i : i \in N, \K\T \subseteq F(i) \}
  \]
  \caption{Modal sequent rules for $\SC_{(N, \cless, F)}$.}
  \label{fig:simply-dep-multimodal}
\end{figure}
Finally, Def.~\ref{def:Tx} of Sec.~\ref{subsec:lbns} can be extended to the multi-modal case in a trivial way, resulting in the labelled nested system $\LLNS_{(N,\cless,F)}$ (Fig.~\ref{fig:lab-indnes}).
\begin{figure}[t]
$
\begin{array}{c}
\infer[\mathbb{TL}({\ibox{ij}}_L^n)]{\Rscr,xR_jy, X, x\lab \ibox{i}A \seq Y}{\Rscr,xR_jy, X, x\lab \ibox{i}A,y\lab A \seq Y}
\quad
\infer[\mathbb{TL}({\ibox{i}}_R^n)]{\Rscr,X\seq Y, x:\ibox{i}A}{\Rscr,xR_iy,X\seq Y,y\lab A}
\quad
\infer[{\mathbb{TL}(\mathsf{t}_i)}]{\Rscr, X, x\lab \ibox{i}A \seq Y}{\Rscr, X, x\lab \ibox{i}A,x\lab A \seq Y}
\\\\
\infer[\mathbb{TL}({\mathsf{d}_{ij}}^n)]{\Rscr, X, x\lab \ibox{i}A \seq Y}{\Rscr,xR_jy, X, x\lab \ibox{i}A,y\lab A \seq Y}
\qquad
\infer[\mathbb{TL}(\4_{ij}^n)]{\Rscr,xR_jy, X, x\lab \ibox{i}A \seq Y}{\Rscr,xR_jy, X, x\lab \ibox{i}A,y\lab \ibox{i}A \seq Y}
\end{array}
$
\caption{Modal rules for labelled indexed nested system $\LLNS_{(N,\cless,F)}$.}\label{fig:lab-indnes}
\end{figure}

\begin{theorem}\label{prop:equiv2}
Systems $\NS_{(N,\cless,F)}$, $\LNS_{(N,\cless,F)}$, $\SC_{(N,\cless,F)}$ and $\LLNS_{(N,\cless,F)}$  are equivalent.
\end{theorem}
By straightforwardly extending the geometric rule scheme presented in~\cite{Negri:2005p878}  for normal modalities to the multi-modal case, we can propose
$\Gtmm$, a sound and complete labelled sequent system for $(N,\cless,F)$.
Figs.~\ref{fig:semantic} and~\ref{fig:relmulti} present the modal and relational rules of
$\Gtmm$ (see also~\cite{DBLP:conf/lics/GargGN12}). 

\begin{figure}[t]
\begin{subfigure}{.4\linewidth}
$
\begin{array}{c}
   \infer[\ibox{i}_L^t]{\Rscr,xR_iy,\world{x}{\Box_i A},  \Gamma \seq\Delta}{\Rscr,xR_iy,\world{x}{\Box_i A}, \world{y}{A}, \Gamma \seq\Delta}
\\ \\ \infer[\ibox{i}_R^t]{\Rscr,\Gamma \seq\Delta, \world{x}{\Box_i A}}{\Rscr,xR_i y,\Gamma \seq\Delta, \world{y}{A}}
\end{array}
$
\caption{Modal rules.}
\label{fig:semantic}
\end{subfigure}
\begin{subfigure}{.65\linewidth}
$
\begin{array}{cc}
\infer[\text{Ref}_i]{\Rscr,\Gamma \seq\Delta}{\Rscr,xR_ix, \Gamma \seq\Delta} & 
  \infer[\text{Trans}_i]{\Rscr,xR_iy, yR_iz, \Gamma \seq\Delta}{\Rscr,xR_iz, xR_iy, yR_iz, \Gamma
\seq\Delta} \\[5pt]
\infer[\text{Ser}_i]{\Rscr,\Gamma \seq\Delta}{\Rscr,xR_iy, \Gamma \seq\Delta} &
\infer[\text{Int}\; {[j\cless i]}]{\Rscr,xR_jy,\Gamma \seq\Delta}{\Rscr,xR_jy,xR_iy, \Gamma \seq\Delta}  \\[5pt]
\end{array}
$
\caption{Multi-modal relational rules. $y$ is  fresh in \text{Ser}.}
\label{fig:relmulti}
\end{subfigure}
\caption{Some rules of the labelled system $\Gtmm$.}
\end{figure}

The next results follow the same lines as the ones in Sec~\ref{sec:int}.
\begin{lemma}\label{lemma:mml}
The rules $\mathbb{TL}(\ibox{ij}_L),\mathbb{TL}(\mathsf{d}_{ij}^n),\mathbb{TL}(\mathsf{t}_i^n), \mathbb{TL}(\4_{ij}^n)$ are derivable  in $\Gtmm$.
\end{lemma}
\begin{proof} Suppose $i\preceq j$.
$\mathbb{TL}(\ibox{ij}_L^n)$ is derivable as
$$
\infer[\text{Int}]{\Rscr,xR_jy, X, x\lab \ibox{i}A \seq Y}{
\infer[\ibox{i}_L^t]{\Rscr,xR_jy,xR_iy, X, x\lab \ibox{i}A \seq Y}
{\Rscr,xR_jy,xR_iy, X, x\lab \ibox{i}A,y\lab A \seq Y}}
$$
Derivability of $\mathbb{TL}(\mathsf{d}_{ij}),\mathbb{TL}(\mathsf{t}_i), \mathbb{TL}(\4_{ij})$ are also straightforward. For example, if $j\preceq i$ and $\K\D\subseteq F(i)$, 
$$
\infer[\text{Ser}_i]{\Rscr, X,x:\Box_i A\seq Y}
{\infer[\text{Int}]{\Rscr, X,xR_j y,x:\Box_i A\seq Y}
{\infer[\ibox{i}_L^t]{\Rscr, X,xR_j y,xR_ iy,x:\Box_i A\seq Y}
{\deduce{\Rscr, X,xR_j y,xR_ iy,x:\Box_i A, y:A\seq Y}{}}}}
$$
 \qed
\end{proof}
\begin{theorem}
$\Gtmm$ is  label-preserving equivalent  to $\LLNS_{(N,\cless,F)}$.
\end{theorem}
\begin{proof}
That every provable sequent in $\LLNS_{(N,\cless,F)}$ is provable in  $\Gtmm$ is a direct consequence of Lemma~\ref{lemma:mml}. For the other direction, 
observe that rule relational rules can be restricted so to be applied just before a  $\ibox{i}_L^t$ rule. 
 \qed
\end{proof}
This means that labels in $\NS_{(N,\cless,F)}$ represent worlds in a {\em $(N,\cless,F)$-frame}, and this extends the results in~\cite{DBLP:conf/aiml/GoreR12} for the case of multi-modality. 
 
\section{Non-normal modal systems}\label{sec:nn}
We now move our attention to non-normal modal logics, i.e., modal logics that are not extensions of $\K$. In this work, we will consider the \emph{classical modal logic} $\mathsf{E}$ and the {\em monotone modal logic} $\M$ . Although our approach is general enough for considering nested, linear nested and sequent systems for several extensions of such logics (such as the {\em classical cube} or the {\em modal tesseract} -- see~\cite{DBLP:journals/corr/LellmannP17}), there are no satisfactory labelled sequent calculi in the literature for such extensions.

%
%
For constructing nested calculi for these logics, the sequent rules should be decomposed into their different components. However, there are two complications compared to the case of normal modal logics: the need for (1) a mechanism for capturing the fact that exactly one boxed formula is introduced on the left hand side; and (2) a way of handling multiple premises of rules. The first problem is solved by introducing the indexed nesting  $\nese{\cdot}$ to capture a state where a sequent rule has been partly processed; 
the second problem is solved by making the nesting operator $\nese{\cdot}$ {\em binary}, which permits the storage of more information about the premises. 
Fig.~\ref{fig:nesnn} presents a unified nested system for logics $\NS_{\mathsf{E}}$ and $\NS_{\M}$.

\begin{figure}[t]
\[
\infer[\Box_R^{\e n}]{\Gamma \seq\Delta,\Lambda,\Box B}{\Gamma \seq\Delta,\Lambda,\nese{\;\seq B;B\seq\cdot}}
\quad
 \infer[\Box_L^{\e n}]{\Gamma,\Box A \seq\Delta,\Lambda,\nese{\Sigma\seq \Pi;\Omega\seq\Theta}}{\Gamma,\Box A, \seq\Delta,\Lambda,\nes{\Sigma,A\seq \Pi}
 &
 \Gamma,\Box A, \seq\Delta,\Lambda,\nes{\Omega\seq\Theta,A}}
 \]
\[
\infer[\M^n]{\Gamma \seq\Delta,\Lambda,\nese{\Sigma\seq \Pi;\Omega\seq\Theta}}{\Gamma\seq\Delta,\Lambda,\nese{\Sigma\seq \Pi;\Omega,\bottom\seq\Theta}}
 \]
\caption{Modal  rules for systems  $\NS_{\mathsf{E}}$ and  $\NS_\M$.}\label{fig:nesnn}
\end{figure}
$\NS_{\mathsf{E}}$ and $\NS_\M$ are shallow n-directed, fully permutable, with invertible rules and, since propositional rules cannot be applied inside the indexed nestings, the modal rules naturally occur in blocks.  Hence the nested rules can be restricted to the linear version~\cite{DBLP:journals/corr/LellmannP17} (Fig.~\ref{fig:lnsnn}), which in turn correspond to macro-rules equivalent to the sequent rules in Fig.~\ref{fig:em} for $\SC_\mathsf{E}$ and $\SC_\M$.
\begin{figure}[t]
$\qquad\qquad
\begin{array}{c}
  \infer[\Box_R^{\e l}]{\mathcal{G} \lns \Gamma \seq \Box B,
    \Delta}{\mathcal{G} \lns \Gamma \seq \Delta \lnse( \;\seq B; B
    \seq \;)}
  \quad
  \infer[\Box_L^{\e l}]{\mathcal{G} \lns \Gamma,\Box A \seq
    \Delta\lnse(\Sigma \seq \Pi; \Omega \seq \Theta)}{\mathcal{G} \lns \Gamma \seq
    \Delta\lns \Sigma, A \seq \Pi & \mathcal{G} \lns \Gamma \seq
    \Delta\lns  \Omega \seq A, \Theta}
\\\\
  \infer[\mathsf{M}^l]{\mathcal{G} \lnse(\Sigma \seq \Pi; \Omega \seq
    \Theta)}{\mathcal{G} \lnse(\Sigma \seq \Pi; \Omega, \bot \seq
    \Theta)}
\end{array}
$
  \caption{Modal  rules for systems  $\LNS_{\mathsf{E}}$ and  $\LNS_\M$.}
  \label{fig:lnsnn}
\end{figure}
\begin{figure}[t]
$\qquad\qquad\qquad\qquad\qquad
\infer[\mathsf{E}]{\Gamma,\Box A \seq \Box B, \Delta}{ A \seq B & B  \seq A}
\qquad
\infer[\mathsf{M}]{\Gamma,\Box A \seq \Box B, \Delta}{ A \seq B}
$
\caption{Modal sequent rules for non-normal modal logics $\SC_\mathsf{E}$ and $\SC_\M$.}\label{fig:em}
\end{figure}

Finally, using the labelling method in Section~\ref{sec:lbns}, the rules in Fig.~\ref{fig:nesnn} correspond to the rules in Fig.~\ref{fig:lbem}, where $xNy$ and $xN_\e(y_1,y_2)$ are  relation terms.
 
\begin{figure}[t]
 $\begin{array}{c}
  \infer[\mathbb{TL}(\Box_R^{\e n})]{\Nscr, X \seq Y, \world{x}{\Box B}}{\Nscr,xN_\e(y_1,y_2),X,\world{y_2}{B}\seq \world{y_1}{B},Y}
  \quad
  \infer[\mathbb{TL}(\Box_L^{\e n})]{\Nscr,xN_\e(y_1,y_2),\world{x}{\Box A},X \seq Y}
  {\Nscr,xNy_1,\world{y_1}{A},X \seq Y&
  \Nscr,xNy_2,X \seq Y,\world{y_2}{A}}
  \\ \\
 \infer[\mathbb{TL}(\mathsf{M}^n)]{\Nscr,xN_\e(y_1,y_2),X \seq Y}{\Nscr,xN_\e(y_1,y_2),X,\world{y_2}{\bottom}\, \seq Y} 
 \end{array}
  $
\caption{Modal rules for  $\LLNS_\mathsf{E}$ and $\LLNS_\M$  with $y_1,y_2$ fresh in $\Box_R^\e$.}\label{fig:lbem}
\end{figure}
The semantical interpretation of non-normal modalities $\mathsf{E},\M$ can be given
via {\em neighbourhood semantics}, that smoothly extends the concept of Kripke frames in the sense that accessibility relations are substituted by a family of neighbourhoods. 
\begin{definition}\label{def:neig}
A {\em neighbourhood frame} is a pair $\Fscr=(W,N)$ consisting of a set $W$ of worlds and a neighbourhood function $N:W\rightarrow \wp(\wp W)$. A {\em neighbourhood model} is  a pair  $\Mscr=(\Fscr,\Vscr)$, where  $\Vscr$ is a valuation.
We will drop the model symbol when it is clear from the context.
\end{definition}
The truth description for the box modality in the neighbourhood framework is
\begin{equation}\label{boxe}
w\Vdash \Box A\mbox{ iff } \exists X\in N(w).[(X\forcef  A)\wedge (A\lhd X)]
\end{equation}
where $X\forcef A$ is $\forall x\in X.x\Vdash A$ and $A\lhd X$ is $\forall y.[(y\Vdash A)\iimp y\in X]$. The rules for $\forcef$ and $\lhd$ are obtained using the geometric rule approach~\cite{Negri2017} and are depicted in Fig.~\ref{fig:auxnn}.

If the neighbourhood frame is monotonic (i.e. $\forall X\subseteq W$, if $X\in N(w)$ and $X\subseteq Y\subseteq W$ then $Y \in N(w)$), it is easy to see~\cite{Negri2017} that (\ref{boxe}) is equivalent to
\begin{equation}\label{boxm}
w\Vdash \Box A\mbox{ iff } \exists X\in N(w).X\forcef  A.
\end{equation}

\begin{figure}[t]
\begin{subfigure}{.45\linewidth}
$
\begin{array}{c}
\infer[\forcef]{x\in a, a\forcef A, X \seq Y}{x\in a, \world{x}{A}, a\forcef A, X \seq Y}
 \\ \\
\infer[\lhd]{A\lhd a, X \seq Y}{A\lhd a, X \seq Y,\world{z}{A}&
z\in a,A\lhd a,X\vdash Y} 
 \\ \\
\infer[\init^t]{x\in a,X\seq Y,x\in a}{}
\end{array}
$
\caption{Forcing rules, with 
$z$  arbitrary  in $\lhd_L$.}\label{fig:auxnn}
\end{subfigure}
\begin{subfigure}{.55\linewidth}
$
\begin{array}{c}
\infer[\Box_L^{\e t}]{\world{x}{\Box A}, X \seq Y}{a\in N(x), a\forcef A, A\lhd a,X \seq Y}
\\\\
 \infer[\Box_R^{\e t}]{a\in N(x),X \seq Y, \world{x}{\Box B}}{z\in a,a\in N(x),X \seq Y,\world{x}{\Box B}, \world{z}{B}&
\world{y}{B}, a\in N(x),X \seq Y,\world{x}{\Box B}, y\in a
 }
\\\\
\infer[\Box_L^{\m t}]{\world{x}{\Box A}, X \seq Y}{a\in N(x), a\forcef A, X \seq Y}
 \qquad 
 \infer[\Box_R^{\m t}]{a\in N(x),X \seq Y, \world{x}{\Box B}}{a\in N(x),y\in a,X \seq Y, \world{x}{\Box B},\world{y}{B}}
\end{array}
$
\caption{$a$ fresh in $\Box_L^\e,\Box_L^\m$ and $y,z$ fresh in $\Box_R^\e,\Box_R^\m$.}\label{fig:nonnormal}
\end{subfigure}
\caption{Labelled systems $\GtE$ and $\GtM$}
\end{figure}

This yields the set of labelled rules presented in Fig.~\ref{fig:nonnormal}, where
the rules are adapted from~\cite{Negri2017} by collapsing invertible proof steps.
Intuitively,  while the box left rules create a fresh neighbourhood to $x$, the box right rules create a fresh world in this new neighbourhood and move the formula to it. 
\begin{theorem}
$\GtE$ (resp. $\GtM$) is  label-preserving equivalent  to  $\LLNS_{\mathsf{E}}$ (resp. $\LLNS_\M$).
\end{theorem}
\begin{proof}
For sake of readability, we will only show in sequents the principal and auxiliary formulae on the application of rules. 
Let $\pi$ be a normal proof of $\Nscr, X \seq Y$ in $\LLNS_{\mathsf{E}}$. An instance of the blocked derivation
$$
\infer[\mathbb{TL}(\Box_R^{\e n})]{\Nscr, X,\world{x}{\Box A} \seq Y,\world{x}{\Box B}}
{\infer[\mathbb{TL}(\Box_L^{\e n})]{\Nscr,xN_\e(y_1,y_2),X,\world{x}{\Box A},\world{y_2}{B}\seq Y,\world{y_1}{B}}
{\deduce{\Nscr,xNy_1,X,\world{y_1}{A},\world{y_2}{B} \seq Y, \world{y_1}{B}}{\pi_1}&
\deduce{\Nscr,xNy_2,X,\world{y_2}{B} \seq Y,\world{y_1}{B},\world{y_2}{A}}{}}}
$$
is transformed into the labelled derivation
$$
\infer[\Box_L^{\e t}]{X,\world{x}{\Box A} \seq Y,\world{x}{\Box B}}
{\infer[\Box_R^{\e t}]{a\in N(x), a\forcef A, A\lhd a,X \seq Y,\world{x}{\Box B}}
{\infer[\forcef]{y_1\in a, a\forcef A, X \seq Y,\world{y_1}{B}}
{\deduce{\world{y_1}{A},X\seq Y,\world{y_1}{B}}{\pi'}}&
\infer[\lhd]{\world{y_2}{B},A\lhd a,X\vdash Y,y_2\in a}
{\deduce{\world{y_2}{B},X\seq  Y,\world{y_2}{A}}{\pi_2'}&
\infer[\init]{y_2\in a,X\seq Y,y_2\in a}{}}}}
$$
Observe that, in $\pi_1$, the label $y_2$ will no longer be active, hence the formula $\world{y_2}{B}$ can be weakened. The same with $y_1$ in $\pi_2$. Hence $\pi_1/\pi_2$ corresponds to $\pi_1'/\pi_2 '$ and the ``only if'' part holds. The ``if'' part uses similar proof theoretical arguments as in the intuitionistic or (multi) modal case, observing that applications of the forcing rules can be restricted so to be applied immediately after the modal rules.\qed
\end{proof}
Observe that the neighbourhood information is ``hidden'' in the nested approach. That is, creating of a nesting has a two-step interpretation: one of creating a neighbourhood and another of creating a world in it. These steps are separated in Negri's labelled systems, but the nesting information becomes superfluous after the creation of the related world. This indicates that the nesting approach is more efficient proof theoretically speaking when compared to labelled systems. Also, it is curious to note that this ``two-step'' interpretation makes the nested system {\em external}, in the sense that nestings cannot be interpreted inside the syntax of the logic. In fact, it makes use of the $\langle\,]$ modality~\cite{Chellas:1980fk}.

\section{Conclusion and future work}\label{sec:conc}

In this work we gave  sufficient conditions for transforming a nested system into a sequent calculus system, passing through linear nested systems. Moreover, we showed a semantical characterisation  of intuitionistic, multi-modal and non-normal modal systems, via a case-by-case translation between labelled nested to labelled sequent systems. In this way, we closed the cycle of syntax/semantic characterisation for a class of logical systems.

While some of the presented results are expected (or even not new as the semantical interpretation of nestings in intuitionistic or in single-normal modal logics), our approach is, as far as we know, the first done entirely using proof theoretical arguments. Indeed, the soundness and completeness results are left to the tail case of labelled systems, that carry within the syntax the semantic information explicitly. Moreover, we formally established a relationship between sequent and nested systems for a class of logics. Of course this link is not possible for logics that o not have cut-free sequent systems, such as $\Sfi$. This suggests that there should be a relationship between the fact that the nesting rules for such logic are not n-directed and the impossibility of the proposal of an adequate  sequent calculus system.

Finally, this work can be extended in a number of ways, but perhaps the more natural direction is to complete the syntactical/semantical analysis for the 
classical cube~\cite{DBLP:journals/corr/LellmannP17}. This is specially interesting since $\M\mathsf{N}\mathsf{C}=\K$, that is, we should be able to smoothly collapse the neighbourhood approach into the relational one. We observe that nestings play an important role in these transformations, since it enables to modularly build proof systems.

%
 
\bibliographystyle{splncs03}

\end{document}